\documentclass{article}

\usepackage[preprint]{neurips_2026}

\usepackage[utf8]{inputenc} 
\usepackage[T1]{fontenc}    
\usepackage{hyperref}       
\usepackage{url}            
\usepackage{booktabs}       
\usepackage{amsfonts}       
\usepackage{nicefrac}       
\usepackage{microtype}      
\usepackage{xcolor}         
\usepackage{graphicx}
\usepackage{amsmath}
\usepackage{amsthm}
\newtheorem{definition}{Definition}
\newtheorem{assumption}{Assumption}
\newtheorem{proposition}{Proposition}

\title{Structural Bottlenecks on Frequency Representation in End-to-End Audio Models}

\author{%
  Nicole Cosme-Clifford \\
  Yale University\\
  New Haven, CT \\
  \texttt{nicole.cosme@yale.edu} \\
}

\begin{document}

\maketitle

\begin{abstract}
End-to-end neural audio models achieve high-fidelity compression and generation. We might read that performance as evidence they directly represent interpretable features such as pitch and timbre, but a model can produce plausible outputs without doing so. A model may encode these features in any reachable basis, but regardless of which, the features are well described as compositions of time-frequency-localized primitives. Whether state-of-the-art encoders preserve access to these primitives, and thus to compositions of them, remains unclear. Through theoretical analysis and controlled experiments, we show that several state-of-the-art strided convolutional encoders impose two structural bottlenecks, both predictable from architecture and signal structure, on access to these primitives: (1) they collapse primitives into alias equivalence classes, establishing a bound on representational capacity, and (2) they limit the frequency resolution available to learned filters, restricting separability. For well structured data, we find collapse rates of 31-35\% and filter bandwidths 9-35x above the theoretical resolution bound, confirming that both bottlenecks arise under realistic signal conditions. We then introduce Gabor Latent Refactorization (GLRF), a lightweight post-hoc intervention that re-expresses encoder latents in a frequency-localized basis, reducing filter bandwidths from 10–35x to 1.5–3x of the theoretical resolution bound while preserving reconstruction fidelity and improving control over attributes like pitch. These results show that the encoders in question predictably degrade access to frequency-localized primitives, entangling the features that depend on them, and that a lightweight, retraining-free intervention can recover much of that access, improving steerability and interpretability.

\end{abstract}

\section{Introduction}
Modern neural audio systems achieve state-of-the-art performance across compression and generation, with scaling consistently improving empirical results~\cite{defossez2022high, kumar2023high, evans2024fast, evans2025stable}. It is tempting to read that strong performance as evidence that interpretable features like pitch and timbre are directly encoded in these systems' learned representations. Yet, interpretability studies have largely not recovered them ~\cite{singh2025discovering, beguvs2022interpreting, 
lee2017raw, vu2024toward, muckenhirn2019understanding}. Have these models learned interpretable time-frequency structure, or can strong performance exist without it? We argue the latter: state-of-the-art encoders foreclose access to the primitives that underlie these features, and the features do not reappear in some alternative learned form.

Many physical sound-generating systems are, over short timescales, well described by normal modes ~\cite{oppenheim99, morse1986theoretical}. Sound can then be locally approximated as a superposition $x(t)=\sum_i c_i(t)$, where each \(c_i(t)\) is a narrowband oscillation centered at a distinct frequency, \(f_i\). Human auditory perception exploits this structure, as human-meaningful features (e.g., pitch, timbre) follow from relationships between these time-frequency-localized primitives ~\cite{bregman1994auditory}. Whether artificial systems preserve access to these primitives (and compositions of them) is determined in part by encoder architecture.

State-of-the-art models increasingly operate end-to-end over raw waveforms, using strided convolutional encoders to compress high-dimensional signals into compact latents~\cite{defossez2022high, kumar2023high, evans2024fast, evans2025stable, zeghidour2021soundstream, dhariwal2020jukebox}. We show that this class of encoders exhibits two structural bottlenecks on structured control of frequency-localized primitives and that both are architecturally determined. The first, injectivity failure, causes distinct frequency components in the input signal to collapse into indistinguishable proxy components in the encoder's representation space due to downsampling-induced aliasing; we derive an analytical bound on how many components can survive this collapse and confirm it empirically across three state-of-the-art models with correlation $r \approx 0.99$ between predicted and observed collapse rates. The second, separability failure, occurs when components survive collapse but remain inseparable by learned filters, preventing independent manipulation; we show that learned filters operate $10 - 35\times$ above the theoretical resolution limit set by the encoder's receptive field.

This corroborates findings from recent interpretability studies on end-to-end neural audio models, which rarely recover frequency-localized primitives or superpositions of them from learned representations~\cite{singh2025discovering, beguvs2022interpreting, 
lee2017raw, vu2024toward, muckenhirn2019understanding}. It is also consistent with the observation that architectures with explicit frequency structure tend to produce more interpretable representations~\cite{zeghidour2021leaf, ravanelli2018interpretable, 
bruna2013invariant, anden2014deep, caillon2021rave}.

There are practical stakes for this. Structured frequency control could enable music generation systems to manipulate pitch and timbre independently, speech synthesis models to modify prosody without affecting speaker identity, and scientific audio tools to isolate and intervene on specific signal components. However, if the representational interface between the encoder and downstream modules does not permit control over frequency structure, downstream components may not recover it, making the encoder's representational geometry a fundamental constraint on capability.

This work is situated within a growing research program asking how to better utilize the representational budget that existing architectures already provide. Recent work has asked analogous questions of attention mechanisms in LLMs ~\cite{team2026attention, yuan2025native}. In mechanistic interpretability, a parallel tradition asks whether learned representations can be re-expressed in bases that expose interpretable structure without compromising the underlying model~\cite{elhage2022toy, bricken2023monosemanticity}. We ask the same kind of question about strided convolutional audio encoders: given a fixed architecture with a fixed parameterization, is available capacity spent exposing human-meaningful structure? We focus on end-of-encoder latents, as these constitute the interface available to downstream modules. We show that the answer is largely no: the structure is often implicitly present in the representation but inaccessible. 

Our contributions are as follows:

\textbf{Predictive framework.} We identify two structural conditions on primitive accessibility — an injectivity condition governing representational collapse under downsampling, and a separability condition governing frequency resolution — and show that both are predictable from architectural parameters prior to training. We also show that three state-of-the art models (EnCodec, DAC, and Stable Audio) violate both conditions under realistic signal settings. 
    
\begin{itemize}
    \item \textbf{Injectivity.} Given signal structure and stride schedule, the number of injectively representable components can be computed via a forward model. Collapse rates predicted by this model correlate with observed rates at $r \approx 0.99$ across 643 signal configurations. Low factor-complexity sampling rates (22.05, 44.1 kHz) substantially reduce collapse compared to highly composite rates; injectivity failure cannot be corrected post-hoc and requires upstream stride schedule design.

    \item \textbf{Separability.} The receptive field sets a hard upper bound on achievable filter resolution, and a second forward model predicts where learned filters land relative to that bound from kernel geometry alone, without access to learned weights. 
\end{itemize}
    
\textbf{Intervention for separability.} We introduce Gabor Latent Refactorization (GLRF), a lightweight post-hoc method that re-expresses encoder latents in a frequency-localized basis without retraining, reducing filter bandwidths to 1.5--3$\times$ the resolution limit while preserving reconstruction fidelity and improving frequency-structured representational control.

All pretrained models are publicly available, and implementation details are in the appendix.

\section{Injectivity}
\label{sec:injectivity}
We model the encoder as a mapping $E : x(t) \mapsto z \in \mathbb{R}^d$ that compresses the time-domain input into a lower-dimensional latent via a composition 
of convolutional layers, nonlinearities, and strided downsampling. While defined in the time domain, this transformation induces a corresponding transformation over frequency-localized signal components, mapping approximately narrowband components $c_i(t)$ to latent responses $l_i(t)$. For a narrowband component $c_i$ centered at frequency $f_i$, we define its latent proxy $l_i$ as the latent response associated with $c_i$ under spectral analysis of encoder outputs. An encoder is injective with respect to a component set if distinct components map to single, distinguishable proxies. This definition provides a consistent object for evaluating whether component structure is preserved or entangled in the representation. 

\subsection{Theoretical Injectivity Bound}
\label{sec:injectivity-theoretical}
Injectivity is largely constrained by the encoder's downsampling operations, which reduce the effective sampling rate and induce equivalence classes over component frequencies.

\begin{assumption}[Narrowband components]
\label{ass:narrowband}
Each component $c_i(t)$ is sufficiently narrowband that its behavior 
under downsampling is determined by its dominant center frequency $f_i$.
\end{assumption}

Classical aliasing treats downsampling as a local phenomenon: a frequency 
$f$ folds to $f \bmod f_s$ under downsampling at rate $f_s$ ~\cite{oppenheim99}. This 
framing is sufficient for isolated components but obscures the global 
structure that emerges when multiple components are present simultaneously. 
We extend this view to characterize the collective behavior of a component 
set under stacked downsampling.

\begin{definition}[Encoder-induced alias class]
\label{def:alias-class}
Let $E$ be an encoder with effective latent sampling rate $f_s$ and let 
$\{f_i\}_{i=1}^n$ be the center frequencies of a narrowband component 
set. The encoder induces a map $\phi : f_i \mapsto f_i \bmod f_s$ via 
cumulative downsampling. Components $c_i$ and $c_j$ belong to the same 
\emph{encoder-induced alias class} if $\phi(f_i) = \phi(f_j)$. The 
number of distinct alias classes intersecting $\{f_i\}_{i=1}^n$ is 
denoted $q \leq n$.
\end{definition}

\begin{proposition}[Injectivity limit under downsampling]
\label{prop:injectivity}
Under Assumption~\ref{ass:narrowband}, at most $q$ components from 
$\{f_i\}_{i=1}^n$ can be injectively represented. When $n > q$, at 
least two components share an alias class 
(Definition~\ref{def:alias-class}) and become indistinguishable in the 
latent representation.
\end{proposition}

\begin{proof}[Proof sketch]
When $n > q$, the pigeonhole principle guarantees at least two components 
share an alias class and therefore map to the same latent frequency under 
$\phi$. Full proof in Appendix~\ref{app:app_proof_of_prop1}.
\end{proof}

The value of $q$ depends on the arithmetic relationship between 
$f_i$ and $f_s$. For a single component at frequency $f_i$, downsampling at rate $f_s$ maps $f_i$ to $f_i$ mod $f_s$. For a harmonic signal with fundamental $f_0$, the harmonics ${kf_0}$ each fold independently under this map, generating a sequence of folded frequencies ${kf_0}$ mod ${f_s}$. When $f_0/f_s$ is rational with denominator $r$ in lowest terms, this sequence is periodic with period $r$, so all harmonics partition into at most $r$ equivalence classes, yielding $q \leq r$. A small denominator forces many harmonics into few classes and collapse is heavy; a large denominator spreads them out. When $f_i/f_s$ is irrational the sequence does not repeat and the components remain distinct. A third regime arises in practice: when $f_i/f_s$ is close to a rational with small denominator (i.e., slightly detuned harmonic stacks), the orbit does not close mathematically, but folded proxies may land indistinguishably close (this becomes a separability problem). Therefore, near-harmonic stacks may approximate the collapse behavior of harmonic stacks. All regimes are common in practice and injectivity is studied empirically in Section~\ref{sec:injectivity-empirical}. 

\subsection{Empirical Validation of Injectivity Bound}
\label{sec:injectivity-empirical}

\paragraph{Models}
We evaluate whether the collapse predicted by Proposition~\ref{prop:injectivity} 
occurs in practice across three state-of-the-art architectures spanning compression 
and generation (EnCodec, DAC, and Stable Audio), each employing a strided 
convolutional encoder to compress raw audio to latent representations.

We assume these encoders do not learn effective anti-aliasing ~\cite{zhang2019making} prior to downsampling. None of the three architectures constrain or explicitly implement it~\cite{defossez2022high, evans2025stable, kumar2023high}. Our empirical results directly support this: if anti-aliasing were active, observed proxies would not match theoretical predictions, yet we find $r \approx 0.99$ between predicted and observed collapse rates, consistent with classical downsampling theory.

\paragraph{Data}
We construct input signals composed of controlled sets of narrowband components with known center frequencies ${f_i}$ (see Appendix~\ref{app:app_injectivity_analysis}). In Experiment 1, we vary the arithmetic relationship between ${f_i}$ and $f_s$, letting signals span the range of harmonicity and commensurability. In Experiment 2, we fix a single dataset of musically realistic stacks and perform a greedy search over stride schedules, allowing commensurability to vary naturally relative to the local sampling rates induced by these schedules. All experiments use synthetic signals, as measuring collapse requires ground truth component frequencies. Theoretical bounds are independent of this assumption and should hold for real audio. However, empirical results on natural signals may differ due to inharmonicity and broadband content.  

\paragraph{Metrics}
To approximate the encoder-induced equivalence classes, we analyze the spectral 
structure of encoder outputs. For each output, we compute the magnitude spectrum 
per latent channel via FFT and identify dominant peaks via a fixed peak detection procedure (threshold, prominence, and width criteria are detailed in Appendix~\ref{app:app_injectivity_analysis}; results are stable across a range of detection hyperparameters). Each detected peak is treated as a latent proxy corresponding to one or more input components. This provides an approximation 
of equivalence classes, and consistency across metrics (defined below) indicates 
that detected proxies correspond to genuine structural features of the representation.

We estimate $q$ directly from the input frequencies and compare it to the number of detected proxies. Then we quantify empirical injectivity via the collapse rate, defined as the fraction of input components that map to shared latent proxies. A collapse rate of zero indicates that each component is assigned to a distinct proxy, while higher values indicate increasing violations of injectivity. We additionally report (i) mode recovery, the fraction of input components whose frequencies are recovered within very small tolerance, and (ii) energy concentration, measuring the fraction of total energy captured by detected proxies. Note that energy concentration does not approach 1.0, as the encoder introduces additional spectral content through nonlinearities, and this is not accounted for in the predicted alias set. This energy does not affect the validity of the collapse rate metric, which measures only whether predicted components map to distinct proxies.

\paragraph{Results}
Collapse rates are well-predicted by stride schedule and signal frequency structure, independently of learned weights. Across all models and 643 signal configurations, predicted and observed collapse rates correlate at $r \approx 0.99$ (Table~\ref{tab:injectivity}), with residual error attributable to spectral leakage near the peak detector's resolution limit. Results are deterministic given fixed inputs and models; variability in the mode recovery metric arises only from signal configuration, which we control explicitly. 

As predicted by Section~\ref{sec:injectivity-theoretical}, harmonic and near-harmonic configurations with high commensurability exhibit high collapse rates consistent with small $q$, while strongly inharmonic or incommensurate configurations show near-zero collapse. This dependence is directly actionable: highly composite rates (e.g., multiples of 8 kHz) yield substantially higher collapse than rates such as 22.05 kHz or 44.1 kHz, whose factor structure produces fewer alias collisions for typical harmonic signals (Appendix \ref{app:app_injectivity_analysis}). Therefore, given stride schedule $S$ and signal frequencies $\{f_i\}$, the number of injectively representable components $q$ is computed per signal analytically via Proposition~\ref{prop:injectivity}, yielding total predicted collapse rates $\hat{Q} = \frac{n - q}{n}$, independent of learned weights.

\begin{table}[t]
\centering
\begin{tabular}{lrrrrrr}
\toprule
Model & Eff. SR & $\hat{Q}$ & $Q$ & $r$ & Mode Rec. & Energy Conc. \\
\midrule
EnCodec         & 150.00 & 0.311 & 0.349 & 0.990 & $0.962 \pm 0.103$ & $0.640 \pm 0.043$ \\
DAC             & 86.13 & 0.280 & 0.312 & 0.991 & $0.984 \pm 0.072$ & $0.577 \pm 0.046$ \\
Stable Audio    & 21.53 & 0.286 & 0.320 & 0.989 & $0.984 \pm 0.079$ & $0.446 \pm 0.165$ \\
\bottomrule
\end{tabular}
\caption{Injectivity analysis across models. $\hat{Q}$: collapse rate predicted analytically from stride schedule and signal structure alone. $Q$: empirically observed collapse rate. Pearson $r$ measures per-signal correlation across all signal configurations ($n=643$). Mode recovery and energy concentration measure spectral alignment quality (mean ± SD).}
\label{tab:injectivity}
\end{table}

\section{Separability}
\label{sec:separability}
\subsection{Theoretical separability}
\label{sec:theoretical-separability}

An encoder $E$ is separable with respect to a component set $\{c_i\}$ if the latent proxy for each component $c_i$ can be accessed independently. Variations in $c_i$ should induce changes in a localized subset of latent coordinates, with minimal interference from others. While injectivity failure is irreversible, separability failure means components remain present but inaccessible, preventing independent manipulation. Classical signal processing provides two bounds on the frequency resolution required for separability at the level of a single filter~\cite{oppenheim99, morse1986theoretical}. 

First, a single filter of kernel size $K$ at sampling rate $f_s$ resolves frequencies on the order of $\Delta f = f_s/K$. This bound applies to filters treated as isolated operators. However, we show in Section ~\ref{sec:separability-empirical} that learned filters for all models achieve finer resolution than this bound predicts alone. This motivates a second, tighter bound. 

A latent feature cannot depend on input samples outside its receptive field, so the cumulative receptive field of the encoder sets a bound on the temporal support of any feature the encoder can compute. By the uncertainty principle~\cite{oppenheim99, morse1986theoretical}, temporal support $R$ implies frequency resolution no finer than 
\begin{equation}
    \Delta f = \frac{f_s}{R}.
    \label{eq:rf-bound}
\end{equation}
In linear systems, stacked convolutions accumulate resolution toward this bound with receptive field expansion~\citep{oppenheim99}. For the nonlinear systems studied here, we treat this as an upper bound on achievable resolution. So, the receptive field sets the ceiling, but internal operations determine what is realized within it.

\begin{assumption}[Receptive field resolution bound]
\label{ass:receptive-field}
The maximum frequency resolution of a learned filter in a strided 
convolutional encoder is bounded by $\Delta f = f_s/R$, where $R$ is 
the cumulative receptive field.
\end{assumption}

\begin{definition}[Separability ratio]
\label{def:sep-ratio}
Let $B$ denote the bandwidth of a learned filter and $\Delta c$ the 
spacing between adjacent component frequencies. The \emph{separability 
ratio} $\rho = \Delta c / B$ measures whether adjacent components fall 
within the filter's resolution limit. Under Assumption~\ref{ass:receptive-field}, when $\rho < 1$ adjacent components cannot be independently accessed, regardless of injectivity.
\end{definition}

\paragraph{Predictive model for effective filter bandwidth.}
We provide a lightweight model to ballpark, before training, where an encoder's most selective single-lobe filters can land, as these filters are most likely to be selective for individual frequency components. The model uses only kernel sizes, dilations, and stride schedules, with no learned weights and no fitted parameters. 

Each selectivity-carrying layer (kernel size $k > 1$) is modeled as a tractable minimum-uncertainty filter: a Gaussian-tapered impulse response of effective extent $(k_l -1)d_l+1$ (full definition in Appendix \ref{app:sep_sim}). Responses are accumulated by sequential convolution. At each strided stage the response is decimated and re-expressed at the new local sampling rate, so stride contributes through grid rescaling rather than attenuation. Snake activations (used in DAC and StableAudio) are non-saturating and therefore apply no meaningful net attenuation, leaving the main-lobe bandwidth effectively unchanged (validated in Appendix \ref{app:sep_sim}). ELU (used in EnCodec) saturates and contributes a small, bounded broadening (Appendix \ref{app:sep_sim}), although we omit this broadening here to report the purely geometric prediction. 

This is an intentionally simplified structural approximation meant to help anticipate a practical separability ceiling at design time. The results suggest the approximation reliably captures the dominant sources of bandwidth variation: predicted best-case bandwidths fall within 2 Hz of the empirically measured values across all three architectures (Table \ref{tab:separability}). We will devote future work to tightening this model. 

\subsection{Empirical Separability}
\label{sec:separability-empirical}
We evaluate separability by comparing learned filter bandwidths to the resolution limits implied by Section~\ref{sec:theoretical-separability}. For each model, we estimate the frequency response of filters in the encoder's final convolutional layer and measure their bandwidth $B$, restricting analysis to smooth filters with a single dominant lobe (see Appendix~\ref{app:app_separability_analysis}). This provides a conservative estimate of achievable resolution.

Across all models, observed bandwidths fall far above the theoretical receptive field limit, yet they settle near the limit predicted by the 
structural bandwidth model (Appendix \ref{app:sep_sim}; Table~\ref{tab:separability}), confirming the underutilization predicted by Assumption~\ref{ass:receptive-field}. Filters do improve upon the single-kernel bound $\Delta f = \frac{f_s}{K}$, confirming that stacking accumulates some resolution, but the gap to the receptive field bound remains large. As a result, a significant fraction of adjacent components satisfy $\rho = \frac{\Delta c}{B} < 1$, indicating they cannot be independently accessed even when injectivity is preserved.

\begin{table}[t]
\centering
\begin{tabular}{lrrrrr}
\toprule
Model & $K$ & Best BW (Hz) & Best/RF & Pred. (Hz) & Err. (Hz) \\
\midrule
EnCodec      & 7 & 14.16 & 10$\times$ & 16.031 & 1.871 \\
DAC          & 3 &  8.30 & 35$\times$ & 8.248 & 0.052 \\
Stable Audio & 3 &  1.47 & 31$\times$ & 2.063 & 0.593 \\
\bottomrule
\end{tabular}
\caption{Separability analysis. $K$: final-layer kernel size. Best BW: 
narrowest measured filter bandwidth among smooth filters. Best/RF: ratio 
to receptive field bound. Predicted BW: simulated bandwidth from kernel geometry alone. Err: error between simulated bandwidth and measured bandwidth.}
\label{tab:separability}
\end{table}

\begin{figure}[t]
    \centering
    \includegraphics[width=0.6\linewidth]{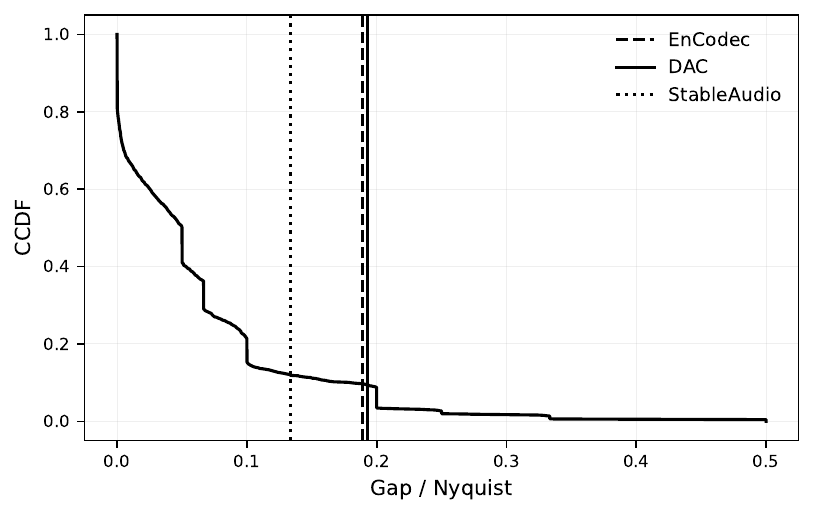}
    \caption{CCDF of filter bandwidths relative to the receptive field resolution limit  $\Delta f = f_s/R$. For typical harmonic signals, adjacent component spacings fall  well within these bandwidths ($\rho = \Delta c / B < 1$), indicating that frequency-adjacent primitives cannot be independently accessed even when injectivity is preserved.}
    \label{fig:CCDF}
\end{figure}

\section{Intervention}
\label{sec:intervention}
Sections~\ref{sec:injectivity} and ~\ref{sec:separability} propose two structural limitations on primitive accessibility: injectivity failure due to downsampling, and limited separability due to insufficient frequency resolution. Injectivity failure depends on the sampling rate and likely cannot be corrected post hoc, so the appropriate fix is upstream stride schedule design. Separability, however, can be improved by re-expressing the latent representation in a frequency-localized basis.

Purpose-built frequency-structured front ends address both bottlenecks by 
design~\cite{choi2016automatic, choi2017convolutional, kong2020panns, 
ravanelli2018interpretable, zeghidour2021leaf, bruna2013invariant, caillon2021rave}. 
However, these require retraining or significant architectural modification and are not the dominant approach. We instead ask what structured control can be recovered from an existing trained encoder without touching its weights or making significant architectural modifications, which is a practically relevant question for deployed systems. The proposed intervention is intended as a proof of concept: we only claim that it recovers meaningful structure within the existing parameter space and provides a concrete starting point for future work.

\subsection{Methods}
\textbf{Gabor Latent Refactorization.} 
We apply a fixed Gabor filterbank to encoder latents $z \in \mathbb{R}^{C \times T}$, parameterized to match the theoretical frequency resolution implied by the receptive field $\Delta f = \frac{f_s}{R}$. The filterbank channels are localized in frequency and tile the frequency spectrum up to Nyquist at the resolution permitted by the receptive field. This produces a decomposition $\tilde{z} \in \mathbb{R}^{C \times 2F \times T}$, where the factor of 2 indexes real and imaginary components. Prior to filterbank application, latents are normalized per-channel by temporal mean and variance. We then learn a channel-wise linear mapping 
\begin{equation}
    M: \mathbb{R}^{C \times 2F \times T} \mapsto \mathbb{R}^{C \times T}
    \label{eq:mapping}
\end{equation}
such that $z \approx M\tilde{z}$ by minimizing
 \begin{equation}
    \left\| z - M\tilde{z} \right\|^2,
    \label{eq:recon-loss}
\end{equation} fit via closed-form ridge regression on a set of synthetic signals. This acts as a change of basis that exposes frequency-localized components without significantly altering the underlying representation. Full filterbank parameters are given in Appendix~\ref{app:app_intervention_details}; computational cost scales linearly with latent size and is negligible relative to the cost of training the encoder and decoder (Appendix~\ref{app:app_compute_resources}).

\textbf{Controllability evaluation. } We evaluate controllability via targeted component substitution, replacing one frequency component in a mixture while holding all others fixed. This is a stricter test of independent component access than interpolation or reconstruction: a representation that only preserves components in entangled form can fail this test even if reconstruction quality is high. For a source component $f_a$ and target component $f_b$ embedded in a mixture, we ask whether an intervention can be identified that reliably (a) reduces energy at the latent proxy for $f_a$ and (b) increases energy at the latent proxy for $f_b$, across varying mixture contexts. We evaluate this across 200 mixture contexts per operation, constructed to avoid latent-frequency collisions between filler components and the source or target component (Appendix~\ref{app:app_intervention_details}). 

In the Gabor-refactorized space, relevant bins are identified analytically from $f_a$ and $f_b$ via the injectivity folding map $\varphi$, and a signed energy change at each bin is read directly from the decomposition. This is meaningful for two reasons. First, the Gabor basis is interpretable by construction: each axis corresponds to a specific narrowband frequency region, so interventions have analytically predictable effects on signal components, subject to the injectivity constraints of Section~\ref{sec:injectivity}, which bound when components remain distinguishable. Second, manipulations in the Gabor basis propagate to the decoded audio via $M$: since $M$ reconstructs the original latent with cosine similarity $\geq 0.97$, the decoder receives a latent very close to those within its trained distribution, and the decoded signal reflects the intended spectral change. Together, these properties mean that GLRF enables structured control over audible full-signal attributes through the existing decoder, without retraining or out-of-distribution perturbation. This amounts to more utilization of signal structure within the existing architecture's parameter space.

In the original latent, no equivalent control procedure exists. An analysis of DAC, for example, reveals that 825 of 1024 channels are broadband and multimodal, with mean main-lobe bandwidth 8.30 Hz (35$\times$ the receptive field bound). This makes it unlikely to identify a small set of channels that reliably tracks a single pitch component independently of its neighbors. Therefore, including all channels in a similarity measure artificially inflates the score since broadband channels activate consistently regardless of which component is present; more sophisticated baselines such as sparse regression face the same problem, as any linear combination remains contaminated by broadband responses. As a baseline, we identify the top-$k$ channels ($k=5$) most responsive to $f_a$ and $f_b$ in isolation and measure whether their energy moves in the correct direction under substitution. A trial is correct if source-associated channels decrease and target-associated channels increase; full details in Appendix~\ref{app:app_intervention_details}.

\subsection{Results}
\textbf{Result 1: Improved Separability.} After re-factorization, filter bandwidths approach the receptive-field bound across all models (Table~\ref{tab:gabor_results}), reducing the gap from $10 - 35\times$ in the original representations to approximately $1.5 - 3\times$. Reconstruction fidelity is preserved throughout, with low log-mel error and high latent cosine similarity, indicating that the learned mapping enables the model to accurately decode the original signal. That a fixed linear mapping reconstructs the original latent with cosine similarity $\geq 0.97$ also implies the encoder's representation is approximately linearly organized with respect to frequency, but not aligned to a basis that supports independent access. GLRF makes this organization explicit. 

To confirm that detected proxies correspond to signal components rather than filterbank atom centers, we also examine the distance between the latent spectral peaks and the nearest Gabor atom center. If GLRF were simply recovering its own discretization grid, peaks would cluster near atom centers. Instead, 91.9\% of peaks fall more than 1 Hz from the nearest atom center, confirming that the Gabor decomposition is tracking signal-determined components, not filterbank artifacts.

Additional results, including interpolation, are in Appendix~\ref{app:app_intervention_details}.

\begin{table}[t]
\centering
\small
\begin{tabular}{lccccc}
\toprule
Model & Log-mel $\downarrow$ & Latent Cosine $\uparrow$ & Min BW (Hz) $\downarrow$ & Min BW / RF $\downarrow$ & Valid (\%) $\uparrow$ \\
\midrule
EnCodec      & 0.62 & 0.98 & 6.93 & 1.62$\times$ & 48\% \\
DAC          & 0.37 & 0.98  & 4.64 & 2.21$\times$ & 95\% \\
Stable Audio & 0.55 & 0.97  & 2.12 & 1.48$\times$ & 94\% \\
\bottomrule
\end{tabular}
\caption{Reconstruction and separability after Gabor refactorization. 
Log-mel measures perceptual reconstruction error; latent cosine measures 
fidelity of the linear mapping to the original latent. Min BW is the 
narrowest recovered filter bandwidth; Min BW/RF reports its ratio to the 
theoretical receptive-field resolution limit. Valid (\%) is the fraction 
of filters satisfying the smoothness criterion. Note that after refactorization, the relevant resolution bound is determined by the Gabor kernel size $K_\text{Gabor}$ rather than the original encoder receptive field, since GLRF extends the encoder with an additional fixed filterbank layer. This is the finest resolution the combined encoder-plus-Gabor system can achieve.}
\label{tab:gabor_results}
\end{table}

\textbf{Result 2: Improved controllability.} We focus on DAC as the strongest case for observing fine-grained frequency control: it exhibits the lowest injectivity failure under typical acoustic signals, the highest proportion of valid frequency-selective filters after refactorization (95\%), and the finest theoretical resolution limit relative to local $f_s$, making it the model where GLRF's bins are potentially most selective and targeted manipulation most interpretable.

\begin{figure}[t]
\centering
\includegraphics[width=\linewidth]{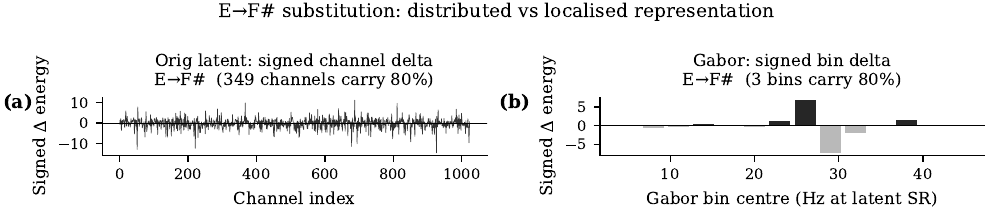}
\caption{Signed energy delta across latent channels (left) and Gabor bins (right) 
for a representative frequency substitution. In the Gabor-refactorized space, 80\% 
of the delta concentrates in 3 bins aligned with the source and target frequencies; 
in the original latent it is distributed across 349 channels with no identifiable 
structure. This localization is what enables reliable targeted manipulation in the 
Gabor space but not in the original latent.}
\label{fig:signed_delta}
\end{figure}

Figure~\ref{fig:signed_delta} illustrates a representative substitution in both spaces. In the Gabor space the delta is concentrated in 3 bins aligned with the source and target frequencies; in the original latent it is distributed across 349 channels with no identifiable structure.

We observe control through two mechanisms. First, latent interpolation between synthetic mixture contexts differing by a single frequency produces clean, crossing trajectories in decoded audio band energy across 300 mixture contexts ($std =0.02\text{--}0.05$, vs. $0.240\text{--}0.270$ in the original latent space), confirming that Gabor-space manipulation has predictable audible consequences in controlled settings. Second, latent interpolation between real acoustic string recordings (NSynth, $n = 1{,}037$ pitch pairs across 48 pitches) produces correct F0 transfer direction in $97.3\%$ of pairs, with correct direction rate increasing from $96\%$ for small intervals (3--5 semitones) to $99\%$ for large intervals (>2 octaves). This is consistent with the separability analysis in Section~\ref{sec:separability-empirical}: larger pitch intervals are spaced by more Gabor bins and are more reliably distinguished in the refactorized space. Linear interpolation additionally produces perceptually consistent pitch transitions across all three models (Appendix~\ref{app:app_intervention_details}), confirming that structured representational geometry is not specific to DAC. 

\textit{Does GLRF enable targeted substitution?} Unlike interpolation-based methods, targeted substitution in the Gabor-refactorized space does not require a paired target endpoint; we can access individual components directly, independently of mixture context. In the Gabor-refactorized space for DAC, substitution succeeds in 200/200 contexts ($100\%$), with the delta vector concentrating $72\%$ of its energy in 5 of 41 bins; the two dominant bins correspond precisely to the source frequency being removed and the target being added. In the original latent, the analogous procedure succeeds in just 60/200 contexts ($30\%$, $95\% CI [0.24, 0.37]$), below the $50\%$ chance baseline. The below-chance rate reflects active interference from context contamination, which we measure directly in Appendix~\ref{app:app_intervention_details}. This confirms that the most frequency-responsive channels are not stably tuned, with their response to a given frequency shifting depending on what other components are present, making them unreliable handles for targeted manipulation. 

\textit{Is the substitution direction stable or context-dependent?} Analogous substitution directions are approximately consistent across mixture contexts in the Gabor space (mean cosine similarity $0.88 \pm 0.09$, $n=200$), confirming they correspond to global directions rather than context-specific perturbations. Consistency is substantially lower in the original latent (mean cosine similarity $0.61 \pm 0.24$), with higher variance confirming that the delta vector shifts significantly depending on mixture context.

\textit{Is the substitution direction frequency-specific or generic?} The Gabor space discriminates between the same operation at different transpositions (intra-cross $=+0.98$), while the original latent largely does not ($+0.31$), confirming that Gabor delta vectors are frequency-specific while latent delta vectors are generally not; the latter tend to capture a global response to change rather than a response to a particular frequency component. 

These results show that encoders preserve linearly separable frequency components 
in entangled forms but largely do not expose them, and that Gabor-space manipulation 
produces predictable audible consequences across synthetic signals, real instrument 
signals, and multiple mixture signals.

\section{Conclusion}
We show that several state-of-the-art strided convolutional encoders impose two structural bottlenecks on access to frequency-localized signal primitives: injectivity failure from downsampling-induced alias collapse, and separability failure from insufficient frequency resolution. Both are systematic consequences of architectural design and signal structure, making failure modes analytically predictable before training. These results suggest that stride schedule is a design decision with measurable representational consequences computable in advance: effective latent sampling rates with low factor complexity (e.g., 22.05 or 44.1 kHz, whose factor structures produce few alias collisions for typical harmonic signals) substantially reduce collapse compared to highly composite rates (e.g., 16, 24, or 48 kHz). Gabor Latent Refactorization demonstrates that recovering frequency structure post-hoc is cheap and effective. 

Future work includes formally developing GLRF into a general-purpose method evaluated across diverse signal domains (e.g., speech, environmental audio, transient-rich content), exploring architectural modifications that jointly address both bottlenecks at design time, and investigating safety applications that utilize the above findings.

To the last point, one primary risk, which is present whenever generative audio improves, is that finer control over perceptual attributes could enhance malicious use cases, like voice cloning or deepfakes. Potential mitigations include signal-dependent watermarking via frequency-localized channels and detection tools grounded in the analytical predictability of the failure modes identified here.

\textbf{Limitations}
Our analysis assumes signals are locally well-approximated as sums of narrowband components. For signals with high component density relative to the effective latent sampling rate or broadband signals, alias classes may grow larger and collapse rates may increase. This may push the model into a regime where the injectivity bound predicts degraded access even before separability is considered. Selectivity works best when the model's receptive field allows the Gabor layer to be selective enough such that most signal component proxies land in different Gabor channels, which makes the results signal-dependent and model-dependent. Most experiments use controlled synthetic signals, since measuring collapse requires ground truth frequencies; empirical rates on real audio may differ, though theoretical bounds should hold regardless. We do not analyze the interaction between injectivity and separability failures, nor test models with explicitly frequency-selective front ends (e.g., LEAF, SincNet). The refactorization addresses separability but not injectivity. Controllability is demonstrated through trajectory-based manipulation and is most consistent in DAC. GLRF is a proof of concept and does not exhaust the space of possible interventions. The linear mapping is fit on simple signals; generalization to broadband or transient-rich content (e.g., percussion, environmental sounds) is not quantitatively assessed and may require mixed training data. 

\bibliographystyle{plainnat}
\bibliography{references}

@article{defossez2022high,
  title={High fidelity neural audio compression},
  author={D{\'e}fossez, Alexandre and Copet, Jade and Synnaeve, Gabriel and Adi, Yossi},
  journal={arXiv preprint arXiv:2210.13438},
  year={2022}
}

@article{kumar2023high,
  title={High-fidelity audio compression with improved rvqgan},
  author={Kumar, Rithesh and Seetharaman, Prem and Luebs, Alejandro and Kumar, Ishaan and Kumar, Kundan},
  journal={Advances in Neural Information Processing Systems},
  volume={36},
  pages={27980--27993},
  year={2023}
}

@inproceedings{evans2025stable,
  title={Stable audio open},
  author={Evans, Zach and Parker, Julian D and Carr, CJ and Zukowski, Zack and Taylor, Josiah and Pons, Jordi},
  booktitle={ICASSP 2025-2025 IEEE International Conference on Acoustics, Speech and Signal Processing (ICASSP)},
  pages={1--5},
  year={2025},
  organization={IEEE}
}

@article{zeghidour2021leaf,
  title={LEAF: A learnable frontend for audio classification},
  author={Zeghidour, Neil and Teboul, Olivier and Quitry, F{\'e}lix De Chaumont and Tagliasacchi, Marco},
  journal={arXiv preprint arXiv:2101.08596},
  year={2021}
}

@article{dhariwal2020jukebox,
  title={Jukebox: A generative model for music},
  author={Dhariwal, Prafulla and Jun, Heewoo and Payne, Christine and Kim, Jong Wook and Radford, Alec and Sutskever, Ilya},
  journal={arXiv preprint arXiv:2005.00341},
  year={2020}
}

@article{zeghidour2021soundstream,
  title={Soundstream: An end-to-end neural audio codec},
  author={Zeghidour, Neil and Luebs, Alejandro and Omran, Ahmed and Skoglund, Jan and Tagliasacchi, Marco},
  journal={IEEE/ACM Transactions on Audio, Speech, and Language Processing},
  volume={30},
  pages={495--507},
  year={2021},
  publisher={IEEE}
}

@article{ravanelli2018interpretable,
  title={Interpretable convolutional filters with sincnet},
  author={Ravanelli, Mirco and Bengio, Yoshua},
  journal={arXiv preprint arXiv:1811.09725},
  year={2018}
}

@article{bruna2013invariant,
  title={Invariant scattering convolution networks},
  author={Bruna, Joan and Mallat, St{\'e}phane},
  journal={IEEE transactions on pattern analysis and machine intelligence},
  volume={35},
  number={8},
  pages={1872--1886},
  year={2013},
  publisher={IEEE}
}

@article{anden2014deep,
  title={Deep scattering spectrum},
  author={And{\'e}n, Joakim and Mallat, St{\'e}phane},
  journal={IEEE Transactions on Signal Processing},
  volume={62},
  number={16},
  pages={4114--4128},
  year={2014},
  publisher={IEEE}
}

@inproceedings{zhang2019making,
  title={Making convolutional networks shift-invariant again},
  author={Zhang, Richard},
  booktitle={International conference on machine learning},
  pages={7324--7334},
  year={2019},
  organization={PMLR}
}

@book{bregman1994auditory,
  title={Auditory scene analysis: The perceptual organization of sound},
  author={Bregman, Albert S},
  year={1994},
  publisher={MIT press}
}

@book{morse1986theoretical,
  title={Theoretical acoustics},
  author={Morse, Philip McCord and Ingard, K Uno},
  year={1986},
  publisher={Princeton university press}
}

@book{oppenheim99,
  author = {Oppenheim, Alan V. and Schafer, Ronald W. and Buck, John R.},
  edition = {Second},
  publisher = {Prentice-hall Englewood Cliffs},
  title = {Discrete-Time Signal Processing},
  year = {1999}
}

@article{singh2025discovering,
  title={Discovering Interpretable Concepts in Large Generative Music Models},
  author={Singh, Nikhil and Cherep, Manuel and Maes, Pattie},
  journal={arXiv preprint arXiv:2505.18186},
  pages={arXiv--2505},
  year={2025}
}

@inproceedings{evans2024fast,
  title={Fast timing-conditioned latent audio diffusion},
  author={Evans, Zach and Carr, CJ and Taylor, Josiah and Hawley, Scott H and Pons, Jordi},
  booktitle={Forty-first International Conference on Machine Learning},
  year={2024}
}

@article{beguvs2022interpreting,
  title={Interpreting intermediate convolutional layers of generative CNNs trained on waveforms},
  author={Begu{\v{s}}, Ga{\v{s}}per and Zhou, Alan},
  journal={IEEE/ACM transactions on audio, speech, and language processing},
  volume={30},
  pages={3214--3229},
  year={2022},
  publisher={IEEE}
}

@article{lee2017raw,
  title={Raw waveform-based audio classification using sample-level CNN architectures},
  author={Lee, Jongpil and Kim, Taejun and Park, Jiyoung and Nam, Juhan},
  journal={arXiv preprint arXiv:1712.00866},
  year={2017}
}

@article{vu2024toward,
  title={Toward end-to-end interpretable convolutional neural networks for waveform signals},
  author={Vu, Linh and Tran, Thu and Lim, Wern-Han and Phan, Raphael},
  journal={arXiv preprint arXiv:2405.01815},
  year={2024}
}

@inproceedings{muckenhirn2019understanding,
  title={Understanding and Visualizing Raw Waveform-Based CNNs.},
  author={Muckenhirn, Hannah and Abrol, Vinayak and Magimai-Doss, Mathew and Marcel, S{\'e}bastien},
  booktitle={Interspeech},
  pages={2345--2349},
  year={2019}
}

@article{choi2016automatic,
  title={Automatic tagging using deep convolutional neural networks},
  author={Choi, Keunwoo and Fazekas, George and Sandler, Mark},
  journal={arXiv preprint arXiv:1606.00298},
  year={2016}
}

@article{kong2020panns,
  title={Panns: Large-scale pretrained audio neural networks for audio pattern recognition},
  author={Kong, Qiuqiang and Cao, Yin and Iqbal, Turab and Wang, Yuxuan and Wang, Wenwu and Plumbley, Mark D},
  journal={IEEE/ACM Transactions on Audio, Speech, and Language Processing},
  volume={28},
  pages={2880--2894},
  year={2020},
  publisher={IEEE}
}

@inproceedings{choi2017convolutional,
  title={Convolutional recurrent neural networks for music classification},
  author={Choi, Keunwoo and Fazekas, Gy{\"o}rgy and Sandler, Mark and Cho, Kyunghyun},
  booktitle={2017 IEEE International conference on acoustics, speech and signal processing (ICASSP)},
  pages={2392--2396},
  year={2017},
  organization={IEEE}
}

@article{team2026attention,
  title={Attention residuals},
  author={Chen, Guangyu and Zhang, Yu and Su, Jianlin and Xu, Weixin and Pan, Siyuan and Wang, Yaoyu and Wang, Yucheng and Chen, Guanduo and Yin, Bohong and others},
  journal={arXiv preprint arXiv:2603.15031},
  year={2026}
}

@inproceedings{yuan2025native,
  title={Native sparse attention: Hardware-aligned and natively trainable sparse attention},
  author={Yuan, Jingyang and Gao, Huazuo and Dai, Damai and Luo, Junyu and Zhao, Liang and Zhang, Zhengyan and Xie, Zhenda and Wei, Yuxing and Wang, Lean and Xiao, Zhiping and others},
  booktitle={Proceedings of the 63rd Annual Meeting of the Association for Computational Linguistics (Volume 1: Long Papers)},
  pages={23078--23097},
  year={2025}
}

@article{elhage2022toy,
  title={Toy models of superposition},
  author={Elhage, Nelson and Hume, Tristan and Olsson, Catherine and Schiefer, Nicholas and Henighan, Tom and Kravec, Shauna and Hatfield-Dodds, Zac and Lasenby, Robert and Drain, Dawn and Chen, Carol and others},
  journal={arXiv preprint arXiv:2209.10652},
  year={2022}
}

@article{bricken2023monosemanticity,
   title={Towards Monosemanticity: Decomposing Language Models With Dictionary Learning},
   author={Bricken, Trenton and Templeton, Adly and Batson, Joshua and Chen, Brian and Jermyn, Adam and Conerly, Tom and Turner, Nick and Anil, Cem and Denison, Carson and Askell, Amanda and Lasenby, Robert and Wu, Yifan and Kravec, Shauna and Schiefer, Nicholas and Maxwell, Tim and Joseph, Nicholas and Hatfield-Dodds, Zac and Tamkin, Alex and Nguyen, Karina and McLean, Brayden and Burke, Josiah E and Hume, Tristan and Carter, Shan and Henighan, Tom and Olah, Christopher},
   year={2023},
   journal={Transformer Circuits Thread},
   url={https://transformer-circuits.pub/2023/monosemantic-features/index.html}
}

@article{caillon2021rave,
  title={RAVE: A variational autoencoder for fast and high-quality neural audio synthesis},
  author={Caillon, Antoine and Esling, Philippe},
  journal={arXiv preprint arXiv:2111.05011},
  year={2021}
}

@inproceedings{engel2017neural,
  title={Neural audio synthesis of musical notes with {WaveNet} autoencoders},
  author={Engel, Jesse and Resnick, Cinjon and Roberts, Adam and Dieleman, 
          Sander and Norouzi, Mohammad and Eck, Douglas and Simonyan, Karen},
  booktitle={International Conference on Machine Learning},
  year={2017}
}

@misc{soxr,
  title={{libsoxr}: The {SoX} Resampler library},
  author={Robinson, Rob},
  url={https://sourceforge.net/projects/soxr/},
  year={2013}
}

\appendix
\section{Technical appendices and supplementary material}
\subsection{Proof of Proposition 1}
\label{app:app_proof_of_prop1}

\paragraph{Setup.}
Let \(x(t)\) be a signal composed of \(n\) approximately narrowband components
\[
x(t) = \sum_{i=1}^n c_i(t),
\]
where each component \(c_i\) has dominant center frequency \(f_i\). Let \(E\) be an encoder containing strided downsampling operations with cumulative downsampling factor \(s\), yielding an effective latent sampling rate \(f_s\). Under Assumption~\ref{ass:narrowband}, each component is sufficiently narrowband that its dominant frequency determines its behavior under downsampling. 

\paragraph{Encoder-induced equivalence classes.}
The downsampling operations induce an equivalence relation \(\sim_E\) over component frequencies, where
\[
f_i \sim_E f_j
\]
if components centered at \(f_i\) and \(f_j\) become indistinguishable after the encoder's downsampling operations. Equivalently, \(f_i\) and \(f_j\) belong to the same encoder-induced alias class.

This relation is reflexive, symmetric, and transitive by construction: every component is indistinguishable from itself; indistinguishability is symmetric; and if two components are each indistinguishable from a third under the same encoder mapping, they are indistinguishable from each other. Therefore, \(\sim_E\) partitions the frequency domain into equivalence classes.

\paragraph{Orbit structure under commensurability.}
Consider a component with frequency \(f\), and let the effective sampling rate be \(f_s\). Repeated application of the aliasing operation induced by downsampling generates a sequence of frequencies
\[
f, \; f \pm f_s, \; f \pm 2f_s, \; \dots
\]
modulo the representable frequency range.

If the ratio \(f / f_s\) is rational, i.e., \(f / f_s = p/r\) in lowest terms, then this sequence is periodic with period \(r\). Thus, the component lies in a finite orbit of size at most \(r\), and only a bounded number of distinct frequencies can be distinguished within this class.

In contrast, if \(f / f_s\) is irrational, the sequence does not repeat, and the orbit does not close, so components remain distinct under repeated aliasing operations.

This explains why commensurate (e.g., harmonic) signals yield a small number of equivalence classes, while incommensurate signals yield a larger number.

\paragraph{Counting bound.}
Let \(q\) be the number of distinct equivalence classes intersecting the set of component frequencies \(\{f_i\}_{i=1}^n\). Since all components within the same class are indistinguishable after encoding, the encoder can assign at most one distinguishable latent proxy to each class. Therefore, the number of distinguishable component proxies is upper bounded by \(q\).

If \(n > q\), then by the pigeonhole principle, at least two components must lie in the same equivalence class. These components cannot be distinguished in the latent representation, so the encoder is not injective with respect to the component set.

Thus, at most \(q\) components can be injectively represented, proving Proposition 1.

\paragraph{Relation to classical aliasing.}
In the special case of single-stage uniform downsampling, \(\sim_E\) reduces to the standard aliasing relation induced by the reduced sampling rate. The formulation above generalizes this view by treating the equivalence classes induced by stacked downsampling as global objects over a multi-component signal, which allows us to count how many components can remain distinguishable.

\subsection{Injectivity Analysis}
\label{app:app_injectivity_analysis}
\begin{figure}[t]
    \centering
    \includegraphics[width=\linewidth]{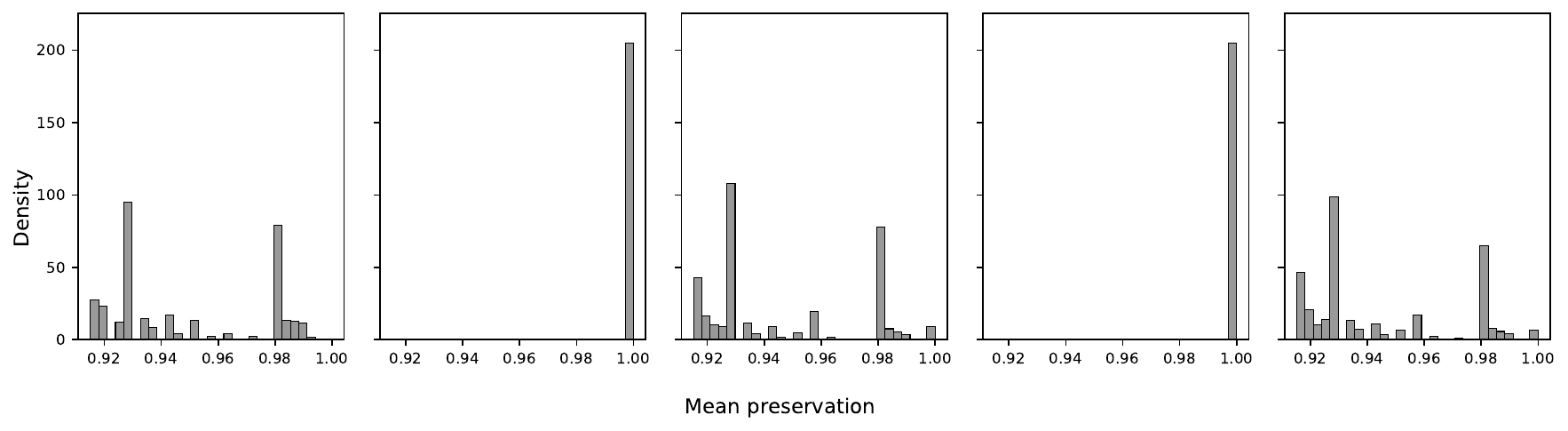}
    \caption{Mean preservation distributions (density) across stride schedules. High factor-complexity rates (16, 24, 48 kHz) exhibit bimodal collapse, while low-complexity rates (22.05, 44.1 kHz) concentrate near perfect preservation.}
    \label{fig:Ideal_sr}
\end{figure}

\textit{Signal construction.} Experiments use five signal configuration types designed to span the harmonic and inharmonic regimes relevant to structured audio.
\begin{itemize}
    \item Harmonic commensurate signals are generated exhaustively: for each $(p,q)$ with $q \in {1,2,3,5,10,15,20}$ and $1 \le p < q$, a fundamental $f_0=\frac{p}{q}$. $f_s$ is generated and a harmonic stack of $n \in {3,5,8,12,15,18,20}$ components built as $f_k=f_0k$. This exhaustively covers configurations where the fundamental is a rational multiple of $f_s$ with small denominator.
    \item Harmonic random uses $f_0 \sim \text{Uniform}[50, 1000]$ Hz, $n\sim \text{Uniform}\{3,\ldots,10\}$, components at $f_k=f_0k$. 
    \item Near-harmonic uses $f_0=f_s(\frac{p}{q}+ \delta)$, $q \in \{2,3,5,10,20\}$, with $p$ drawn uniformly from $\{1,\ldots,q-1\}$ and $\delta \sim \text{Uniform} [-10^{-3}, 10^{-3}]$, constrained to $f_0 \in [50, 1000]$ Hz, with $n \sim \text{Uniform}\{3,\ldots,12\}$.
    \item Inharmonic uses the standard stretched harmonic series $f_0 k \sqrt{1 + Bk^2}$ where $B \sim \text{Uniform}[10^{-4}, 5\times10^{-3}]$ and $f_0 \sim \text{Uniform}[50, 1000]$ Hz, with $n \sim \text{Uniform}\{3,\ldots,10\}$.
    \item Mixture combines 2–3 sub-stacks of randomly selected types with total $n \sim \text{Uniform}\{6,\ldots,20\}$ distributed by multinomial sampling.
\end{itemize}
The harmonic commensurate class is generated exhaustively to ensure complete coverage of all orbit sizes. All signals use base sampling rate 24 kHz, duration 3 s, amplitude normalized by $n$. Alias computation uses tolerance 0.0 (exact aliasing), making collapse rate estimates conservative.

\textit{Activation extraction.} Signals are passed through the full encoder and activations extracted at the end-of-encoder layer. The magnitude spectrum is computed via FFT per latent channel with DC removal. The effective latent sampling rate is tracked through the encoder by accumulating stride factors layer by layer.

\textit{Alias prediction.} For each input frequency $f_i$, the predicted alias is $\phi(f_i) = f_i \bmod f_s$, where $f_s$ is the effective latent sampling rate. Frequencies within tolerance $\tau_\text{theory} = \max(10^{-6},\ \delta f / 2)$ of each other are treated as coincident, where $\delta f$ is the FFT frequency resolution ($\delta f$ is a fraction of 1 Hz for all models). The number of distinct predicted aliases $q$ is then computed from the resulting equivalence classes, and collapse rate is defined as $\frac{n-q}{n}$.

\textit{Mode recovery and energy concentration.} A predicted alias frequency is considered recovered if spectral energy within $2\delta f$ of that frequency exceeds 20\% of the maximum alias energy ($\alpha = 0.2$). $\delta f$ is a fraction of 1 Hz for all models. Mode recovery is the fraction of predicted aliases meeting this criterion. Energy concentration is the fraction of total spectral energy falling within $2\delta f$ of any predicted alias frequency.

\textit{Reported statistics.} Mode recovery and energy concentration are reported as mean $\pm$ SD across signal configurations within each model. Collapse rate is a deterministic function of input frequencies and effective sampling rate, reported as a point estimate per model.

\textit{Stability.} Results are stable across detection tolerances from $0.5\delta f$ to $4\delta f$ and energy thresholds $\alpha \in [0.1, 0.3]$. 

\subsection{Separability Simulation}
\label{app:sep_sim}

\paragraph{Overview}
We simulate the effective filter bandwidth of each encoder's convolutional stack using only kernel sizes, dilations, and stride schedules, with no access to learned weights. The simulation tracks how frequency selectivity accumulates through residual units and is re-expressed at each downsampling stage. Predicted bandwidths are compared against empirically measured best-case filter bandwidths from the sweep described in Appendix~\ref{app:app_separability_analysis}.

\paragraph{Filter Model} 
Given an encoder with $L$ stages, let layer $\ell$ have kernel size $k_\ell$, dilation $d_\ell$, and stride $s_\ell$. Pointwise layers ($k_\ell=1$) are flat in frequency and act as identity $h_\ell = \delta$. However, each selectivity-carrying layer ($k_\ell > 1$) is assigned a best-case impulse response: a Gaussian sampled at the $k_\ell$ dilated tap positions,
\begin{equation}
h_\ell[n] = \exp\!\left(-\frac{(n - \mu_\ell)^2}{2\sigma_\ell^2}\right), \quad n \in \{0, d_\ell, 2d_\ell, \dots, (k_\ell{-}1)d_\ell\},
\end{equation}

with extent $E_\ell = (k_\ell{-}1)d_\ell + 1$, center $\mu_\ell = (E_\ell{-}1)/2$, and width $\sigma_\ell = E_\ell/6$ (the $\pm 3\sigma$ convention); $h_\ell$ is zero at non-tap positions and $\ell_2$-normalized. The cumulative response is built by convolving through the stack, decimating by the stride at each strided stage:

\begin{equation}
g_0 = h_0, \qquad g_\ell = \big(g_{\ell-1} * h_\ell\big)\!\downarrow_{s_\ell},
\end{equation}

where $\downarrow_{s}$ denotes keeping every $s$-th sample. The local sampling rate is tracked alongside, $f_\ell = f_{\ell-1}/s_\ell$, so that decimation rescales the frequency grid without altering the accumulated main-lobe shape. The predicted best-case bandwidth is the $-3$,dB main-lobe width of the final response $g_L$, measured against the final local rate $f_L$:

\begin{equation}
\hat{B} = \text{BW}_{-3\text{dB}}\big(|\mathcal{F}\{g_L\}|, \, f_L\big),
\end{equation} where $\text{BW}_{-3\text{dB}}$ measures bandwidth at the half-power point of the dominant peak. 

\paragraph{Nonlinearities} Snake is treated as bandwidth-preserving: Table~\ref{tab:snake_validation} shows main-lobe bandwidth ratios of $1.000$ between Snake and linear stacks at every encoder stage. ELU introduces a small rectification broadening; we omit it to report the purely geometric prediction. EnCodec is the only architecture of the three to use ELU, and including a broadening term required post hoc fitting but only slightly improved the estimate. We also note that EnCodec's LSTM acts on the latent sequence near the end of the convolutional stack and is not modeled. 

\begin{figure}[t]
    \centering
    \includegraphics[width=\linewidth]{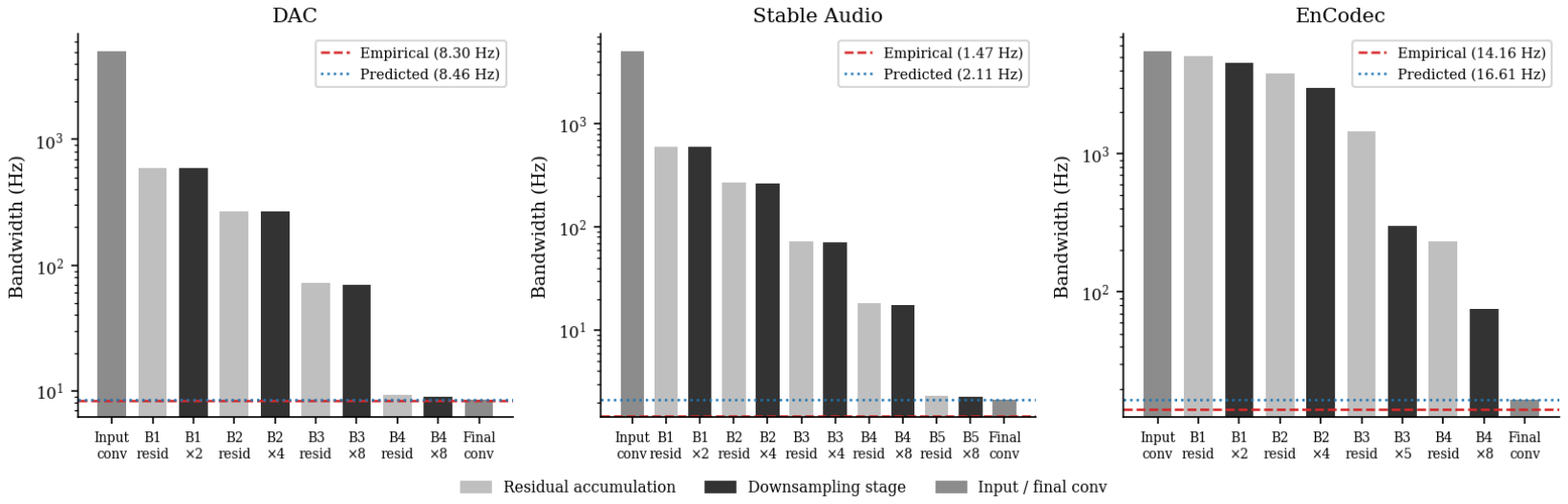}
    \caption{Evolution of filter bandwidth throughout simulated encoder across models.}
    \label{fig:encoder_bw_evo}
\end{figure}

\begin{table}[t]
\centering
\caption{Bandwidth at each encoder stage for Snake vs.\ linear stacks on DAC's architecture. Ratios of 1.000 throughout validate treating Snake as not inhibiting frequency accumulation. Results for Stable Audio are identical and omitted for brevity.}
\label{tab:snake_validation}

\begin{tabular}{lrrr}
\toprule
Stage & Snake BW (Hz) & Linear BW (Hz) & Ratio \\
\midrule
Input conv   & 5512.500 & 5512.500 & 1.000 \\
B1 residuals &  602.930 &  602.930 & 1.000 \\
B1 $\times$2 &  602.930 &  602.930 & 1.000 \\
B2 residuals &  301.465 &  301.465 & 1.000 \\
B2 $\times$4 &  279.932 &  279.932 & 1.000 \\
B3 residuals &   75.366 &   75.366 & 1.000 \\
B3 $\times$8 &   75.366 &   75.366 & 1.000 \\
B4 residuals &    9.421 &    9.421 & 1.000 \\
B4 $\times$8 &    9.421 &    9.421 & 1.000 \\
Final conv   &    9.421 &    9.421 & 1.000 \\
\bottomrule
\end{tabular}
\end{table}

\paragraph{Reproducibility}
The simulation requires only \texttt{numpy} and \texttt{scipy} and runs in under one minute on a standard CPU. No GPU or pretrained model weights are needed. Code is included in the supplementary material. Architecture parameters used are listed in Table~\ref{tab:arch_params}.

\begin{table}[t]
\centering
\caption{Architecture parameters used in the simulation. Residual unit dilations are listed as tuples. EnCodec residual units use no dilation beyond $d=1$.}
\label{tab:arch_params}
\begin{tabular}{llllll}
\toprule
Model & Input SR & Strides & Stride kernels & Res.\ unit $(k, d)$ & Act. \\
\midrule
DAC          & 44100 & 2,4,8,8   & 4,8,16,16  & $(7,1),(7,3),(7,9)$ & Snake \\
Stable Audio & 44100 & 2,4,4,8,8 & 4,8,8,16,16 & $(7,1),(7,3),(7,9)$ & Snake \\
EnCodec      & 48000 & 2,4,5,8   & 4,8,10,16  & $(3,1),(1,1)$       & ELU   \\
\bottomrule
\end{tabular}
\end{table}

\subsection{Separability analysis}
\label{app:app_separability_analysis}
\textit{Filter sweep.} Frequency selectivity is measured by sweeping single-frequency sinusoidal inputs through each encoder and measuring the response of each latent channel. These inputs were chosen as a reliable baseline for selectivity in the most pure signal case. For each model, 300 frequencies are swept linearly from 1 Hz to the Nyquist frequency of the effective latent sampling rate $\frac{f_s}{2}$. Input signals are generated at the model's native sampling rate with duration 2 seconds. Responses are measured as mean squared activation per channel, producing a frequency response curve per channel.

\textit{Bandwidth computation.} Frequency responses are folded to the local frequency axis via $f\bmod f_s$ with reflection at $\frac{f_s}{2}$ to account for aliasing. Each folded response curve is lightly smoothed with a Gaussian filter ($\sigma = 2$ bins) to suppress noise before peak detection. Bandwidth is measured at the half-power point (-3 dB, `frac=0.5`) of the dominant peak, defined as the tallest peak with prominence exceeding 30\% of the maximum response and width of at least 3 bins. Edge bins (3 on each side) are excluded from peak detection to avoid boundary artifacts. Bandwidth is reported in Hz normalized by $f_s$. 

\textit{Filter validity criterion.} A filter is considered valid (and included in the Best BW statistic) if its folded response has exactly one dominant peak with peak dominance exceeding 0.7, where dominance is defined as the height of the tallest peak divided by the sum of all peak heights. This criterion selects filters with unimodal, frequency-selective responses and excludes broadband or multimodal filters that do not admit a meaningful bandwidth measurement. The Valid (\%) column in Table~\ref{tab:separability} reports the fraction of filters meeting this criterion per model.

\textit{Resolution bounds.} The single-kernel bound is $\Delta f = \frac{f_s}{K}$, where $K$ is the kernel size of the final convolutional layer. The receptive field bound is $\Delta f = \frac{f_s}{R}$, where $R$ is the cumulative receptive field computed as $R = 1 + \sum_{\ell=1}^{L}(k_\ell - 1)\prod_{j=1}^{\ell-1} s_j$, with $k_\ell$ the kernel size and $s_j$ the stride at layer $j$ (Eq.~\ref{eq:rf-bound}). Best BW is the minimum bandwidth observed among valid filters. and Best/RF report the ratios of Best BW to the single-kernel and receptive field bounds respectively.

\textit{Reported statistics.} Best BW, Best/RF, and Valid (\%) are deterministic given fixed model weights and sweep parameters, reported as point estimates per model. 

\subsection{Intervention details}
\label{app:app_intervention_details}
\textit{Gabor filterbank.} A fixed Gabor filterbank is applied to the encoder latents $z \in \mathbb{R}^{C \times T}$. The filterbank consists of complex Gabor filters with center frequencies spaced linearly. Each filter uses a Hann-windowed sinusoid with kernel size $K$ samples, producing real and imaginary components. Filterbank center frequencies, spacing, and kernel size are tuned per model to match the effective latent sampling rate of each encoder; parameters for each model are listed in Table~\ref{tab:gabor_params}. These were chosen with respect to the encoder's full receptive field and its given selectivity bound, $\Delta f = \frac{f_s}{R}$ ~\ref{eq:rf-bound}. Center frequencies are normalized to $[0, 0.5]$ by dividing by $f_s$. Filters are $\ell_2$ normalized before application. The filterbank output $\tilde{z} \in \mathbb{R}^{C \times 2F \times T}$ concatenates real and imaginary responses, where $F$ is the number of filter frequencies.

\begin{table}[t]
\centering
\small
\begin{tabular}{lcccc}
\toprule
Model & $f_s$ (Hz) & Freq. Range (Hz) & Spacing (Hz) & Kernel Size \\
\midrule
EnCodec      & 150   & 1–75 & 5 & 35 \\
DAC          & 86.13 & 1–45 & 3 & 41 \\
Stable Audio & 21.53 & 1–11 & 2 & 15 \\
\bottomrule
\end{tabular}
\caption{Parameters for Gabor intervention across models.}
\label{tab:gabor_params}
\end{table}

\textit{Latent normalization.} Prior to filterbank application, latents are normalized per-channel: $z \leftarrow (z - \mu) / (\sigma + 10^{-6})$, where $\mu$ and $\sigma$ are the temporal mean and standard deviation of each channel. This ensures the linear mapping operates on a consistent scale across models.

\textit{Linear mapping.} A channel-wise linear mapping $\in \mathbb{R}^{C \times 2F}$ (Eq. ~\ref{eq:mapping}) is learned to reconstruct the original latent from the Gabor decomposition, minimizing $\| z - M\tilde{z} \|^2$ (Eq. ~\ref{eq:recon-loss}) per channel via ridge regression:
\begin{equation}
w_c = (X_c^\top X_c + \lambda I)^{-1} X_c^\top z_c
\label{eq:ridge}
\end{equation}
where $X_c \in \mathbb{R}^{NT \times 2F}$ is the Gabor decomposition for channel $c$ stacked across all $N$ training signals, $z_c \in \mathbb{R}^{NT}$ is the corresponding concatenated latent activations, and $\lambda = 10^{-4}$ is the ridge penalty. This is a closed-form solution requiring no iterative optimization.

\textit{Training signals.} The mapping is fit on 1000 harmonic signals with fundamental frequencies spaced linearly from 100 to 500 Hz, each with 4 harmonics, at the model's native sampling rate with duration 2 seconds. Signals are peak-normalized before encoding.

\textit{Reconstruction evaluation.} Reconstructed latents $\hat{z} = M\tilde{z}$ are passed through the model's quantizer and decoder to produce audio. Reconstruction fidelity is measured via log-mel distance computed with an 80-band mel spectrogram (FFT size 1024, hop length 256, sample rate 44100 Hz) and latent cosine similarity between $z$ and $\hat{z}$. 

\textit{Separability evaluation after refactorization.} The combined Gabor-plus-linear model is treated as a new encoder and analyzed using the same bandwidth sweep procedure described in Appendix~\ref{app:app_separability_analysis}. Filter validity, Best BW, and Best/RF are computed identically to the baseline analysis, allowing direct comparison.

\subsubsection{Controllability evaluation details}
\textit{Dataset construction.} We test on synthetic musical data, which provides a range of easily testable compositions of frequency components with clear semantic meaning. Mixture contexts are constructed by building three-component harmonic stacks over a two-octave chromatic root grid (C4–B5). A context is retained only if every component aliases to a Gabor bin at least 2 bins away from both the source and target bin under the folding map $\phi$. This collision exclusion ensures that filler components do not contaminate the source or target frequency proxies. For each operation, 200 contexts are sampled; where fewer unique clean contexts exist, sampling with replacement is used and noted. Operations are additionally required to have source and target bins separated by at least 2 bins; operations failing this criterion are excluded entirely.

\textit{Targeted substitution procedure.} For each context, a source chord and target chord are synthesized as sums of pure sinusoids, peak-normalized. Both are passed through the DAC encoder. In the Gabor-refactorized space, the signed energy delta at each bin is computed as $\Delta e_b = e_b^{\text{tgt}} - e_b^{\text{src}}$, where $e_b$ is mean bin energy over time. A trial is correct if $\Delta e_{b_s} < 0$ and $\Delta e_{b_t} > 0$, where $b_s$ and $b_t$ are the bins nearest the folded source and target frequencies respectively. In the original latent, the top-k channels most responsive to the source and target frequencies in isolation ($k = 5$) are identified by encoding single-frequency sinusoids. A trial is correct if the mean signed energy change over source-associated channels is negative and over target-associated channels is positive.

\textit{Consistency evaluation.} Profile cosine similarity is computed pairwise across all 200 context delta vectors within each space, using the full-dimensional signed delta profile (all 41 Gabor bins or all 1024 latent channels). Mean and standard deviation are reported over all $\binom{200}{2} = 19{,}900$ pairs. Transposition discrimination is computed by generating a matched transposed dataset (source and target frequencies multiplied by $5/4$, contexts collision-checked independently), computing mean intra-group cosine for original and transposed operations separately, and reporting the difference between mean intra-group and mean cross-group cosine.

\begin{figure}[t]
\centering
\includegraphics[width=\linewidth]{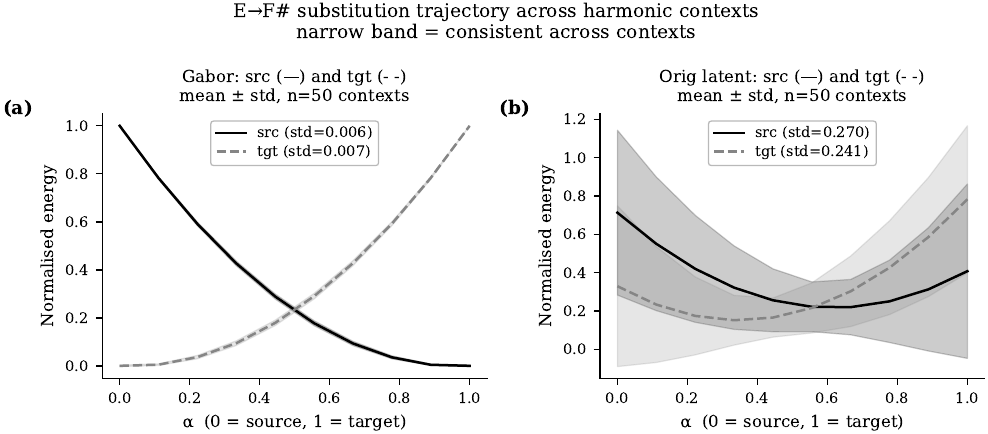}
\caption{Normalized source and target energy cross smoothly along the substitution trajectory in Gabor space, but not in latent space.}
\label{fig:latent-trajectories}
\end{figure}

\begin{figure}[t]
\centering
\includegraphics[width=\linewidth]{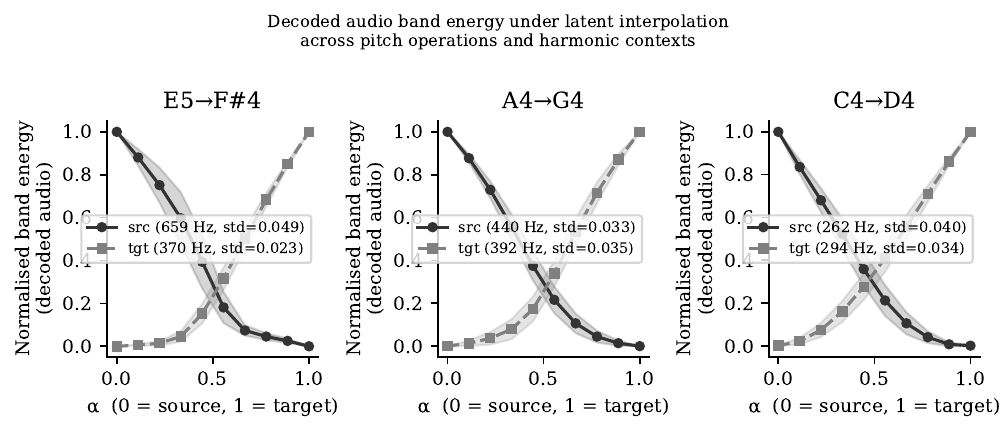}
\caption{Normalized source and target energy crossing smoothly in latent space translates to smooth motion along the substitution trajectory in output space.}
\label{fig:Output-trajectories}
\end{figure}

\begin{figure}[t]
    \centering
    \includegraphics[width=\linewidth]{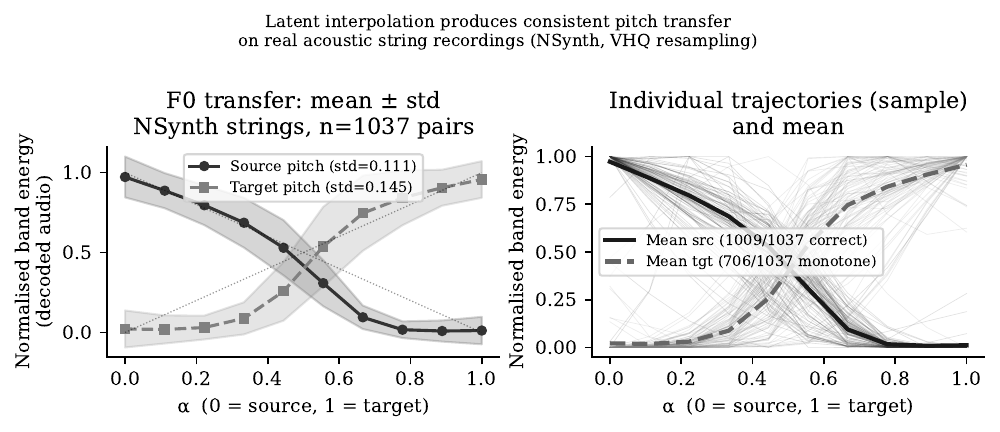}
    \caption{F0 transfer under latent interpolation on NSynth acoustic 
    string recordings ($n=1{,}037$ pitch pairs, 48 pitches, MIDI 
    36--84). \textit{Left:} Mean $\pm$ std of normalised band energy 
    trajectories. \textit{Right:} Individual pair trajectories (sample 
    of 100) with mean overlaid. Correct F0 transfer direction is 
    achieved in 97.3\% of pairs.}
    \label{fig:nsynth_f0}
\end{figure}

\subsection{Decoded Audio Controllability}
\label{sec:audible}

\paragraph{Synthetic chord evaluation.}
Figure~\ref{fig:latent-trajectories} shows normalized source and target energy 
along the substitution trajectory, with mean $\pm$ std across contexts. 
Gabor bins show tight, consistent crossing trajectories: energy smoothly 
disappears from the source bin while smoothly appearing in the target bin 
across all transformation contexts. Original latent channels show high 
variance and failure to differentiate source from target energy: both 
trajectories follow the same u-shape and variance is high, meaning the 
representation mounts a global signal response to the transformation 
rather than a change in one component alone.

To confirm that Gabor-space manipulation produces predictable audible 
consequences, we evaluate latent interpolation between chord contexts 
differing by a single pitch component across three pitch operations 
(C4$\to$D4, A4$\to$G4, E5$\to$F\#4) and 100 harmonic contexts each. 
For each context, we encode source and target chords through DAC, 
interpolate linearly between their latents ($\alpha \in [0,1]$), decode 
each interpolated latent, and measure band energy around the source and 
target pitch frequencies in the decoded audio. 
Figure~\ref{fig:Output-trajectories} shows that source pitch energy 
decreases monotonically while target pitch energy increases across all 
three operations, with standard deviations of $0.02$--$0.05$ — 
substantially below those of the original latent space ($0.24$--$0.27$, 
Fig.~\ref{fig:latent-trajectories}). This confirms that the structured 
geometry of the Gabor-refactorized space propagates faithfully through 
the decoder to the audible signal.

\paragraph{Real acoustic recordings.}
To evaluate generalization beyond synthetic signals, we test latent 
interpolation on real acoustic string instrument recordings from the 
NSynth dataset~\cite{engel2017neural}. We select acoustic string samples 
in the pitch range MIDI 36--84, retaining one recording per pitch and 
excluding the highest pitches where FFT-based F0 estimation is 
unreliable. This yields 48 pitches; all pairs with a minimum gap of 3 
semitones are evaluated, giving 1,037 pitch pairs. Recordings are 
natively 16\,kHz and resampled to 44.1\,kHz using SoXR very-high-quality 
resampling~\cite{soxr}; resampling introduces no measurable F0 error 
relative to scipy-based resampling (mean error $<1$\,Hz on all samples).

For each pair, we encode source and target recordings through DAC, 
interpolate linearly between their latents, decode each step, and measure 
band energy around each pitch's fundamental frequency. A trial is 
classified as correct if source-pitch energy at $\alpha=1$ is lower than 
at $\alpha=0$ and target-pitch energy at $\alpha=1$ is higher than at 
$\alpha=0$.

\begin{table}[t]
\centering
\small
\caption{Correct F0 transfer direction by interval size across 1,037 
NSynth acoustic string pitch pairs. Correct direction rate increases with 
interval size, consistent with larger intervals spanning more Gabor bins 
and being more reliably distinguished in the refactorized space.}
\label{tab:nsynth_results}
\begin{tabular}{lrr}
\toprule
Interval & Pairs & Correct direction \\
\midrule
3--5 semitones   & 129  & 124/129 \hspace{1em}(96\%) \\
6--11 semitones  & 231  & 226/231 \hspace{1em}(98\%) \\
12--23 semitones & 359  & 344/359 \hspace{1em}(96\%) \\
$>$24 semitones  & 318  & 315/318 \hspace{1em}(99\%) \\
\midrule
Total            & 1037 & 1009/1037 \hspace{0.5em}(97.3\%) \\
\bottomrule
\end{tabular}
\end{table}

Results are shown in Table~\ref{tab:nsynth_results} and 
Figure~\ref{fig:nsynth_f0}. Correct F0 transfer direction is achieved in 
$97.3\%$ of pairs. The correct direction rate increases monotonically 
with interval size, from $96\%$ for small intervals (3--5 semitones) to 
$99\%$ for large intervals ($>$2 octaves). This is consistent with the 
separability analysis in Section~\ref{sec:separability}: larger pitch 
intervals span more Gabor bins and are more reliably distinguished in the 
refactorized space, while small intervals near the Gabor resolution limit 
are harder to separate. The 28 incorrect-direction pairs are 
concentrated in the smallest interval category, further supporting this 
interpretation.

\subsubsection{Interpolation}
Linear interpolation between two tones (440 Hz and 660 Hz) in the Gabor-refactorized space produces smooth, perceptually consistent pitch transitions after decoding across all three models (Table~\ref{tab:interp_results}). Trajectories are monotone in perceptual distance from each endpoint and approximately uniform in latent step size. These results confirm that the refactorized space supports structured trajectories, though we note that interpolation acts on the whole signal and does not require component-level separability; it is provided as a sanity check on representational continuity rather than as a controllability result.

\begin{table}[t]
\centering
\small
\begin{tabular}{lccc}
\toprule
Model & Monotone & CV & Endpoint Distance \\
\midrule
EnCodec      & True & 0.289 & 1.3984 \\
DAC          & True & 0.242 & 2.0346 \\
Stable Audio & True & 0.277 & 1.4068 \\
\bottomrule
\end{tabular}
\caption{Interpolation results across models. All three models achieve monotone trajectories. CV measures uniformity of latent step sizes (lower = more linear). Endpoint distance measures the distance from start to end.}
\label{tab:interp_results}
\end{table}

\subsection{Compute Resources}
\label{app:app_compute_resources}
All experiments were conducted on a single NVIDIA A100 40GB GPU. No large-scale model training was performed; the only learned component is the GLRF linear mapping, which is fit via closed-form ridge regression. 

Runtimes were as follows. Injectivity analysis required approximately 1–4 minutes per model depending on architecture (total across all models: $\sim$7 minutes). Separability analysis required approximately 17 seconds to 4 minutes per model (total: $\sim$6 minutes). GLR fitting required under 1 minute total across all models. Evaluation required $\sim$15 minutes. All experiments combined required under 30 minutes of GPU time.

Pretrained model weights were downloaded from public repositories; download times are not included above. Preliminary experiments to create datasets, validate parameters and Gabor filterbank parameterization required approximately 3-4 additional hours of exploratory compute. Given the minimal compute requirements described here, the environmental impact of this work is negligible: by avoiding retraining with GLRF, we save environmental cost.

\subsection{Assets and Licenses}
\textit{Pretrained models.} The following pretrained encoder models are used in our analysis:

\begin{table}[t]
\centering
\small
\begin{tabular}{llll}
\toprule
Model & Version & License & URL \\
\midrule
EnCodec      & 48 kHz & MIT & github.com/facebookresearch/encodec \\
DAC          & 44.1 kHz & MIT & github.com/descriptinc/descript-audio-codec \\
Stable Audio & 44.1 kHz & Community License Agreement & https://huggingface.co/stabilityai/stable-audio-open-1.0 \\
\bottomrule
\end{tabular}
\caption{Pretrained models used in this work, with versions and licenses.}
\label{tab:licenses}
\end{table}

\textit{Released assets.} Code released with this paper is provided under the MIT License. The synthetic signal datasets used in our experiments are fully reproducible from the generation code included in the repository and require no additional data access.

\end{document}